\documentclass[11pt,a4paper]{article}
\usepackage[latin1]{inputenc}
\usepackage{amsfonts}
\usepackage{amscd}
\usepackage{amsmath}
\usepackage{theorem}
\usepackage{mathrsfs}
\usepackage{graphicx}
\usepackage{fancybox,amssymb}
\usepackage{geometry}
\usepackage[english]{babel}
\usepackage{indentfirst}
\usepackage{epstopdf}
\usepackage{authblk}
\usepackage{hyperref}

\newcommand{\R}{\mathbb{R}}
\newcommand{\N}{\mathbb{N}}
\newcommand{\mP}{\mathbb{P}}
\newcommand{\PP}{\mathbb{P}_{(t,x)}}
\newcommand{\e}{{\mathlarger{\rm e}}}
\newcommand{\mE}{\mathbb{E}}

\newcommand{\E}{\mathbb{E}}
\newcommand{\md}{\;{\rm d}}
\newcommand{\s}{\sum\limits}
\newcommand{\de}{\Delta}

\newcommand{\w}{\wedge}

\newcommand{\bq}{\begin{eqnarray*}}
\newcommand{\eq}{\end{eqnarray*}}
\newcommand{\one}{1\mkern-5mu{\hbox{\rm I}}}

\theoremstyle{break}
\theorembodyfont{}
\newtheorem{Def}{Definition}[section]
\newtheorem{Bem}[Def]{Remark}
\newtheorem{Lem}[Def]{Lemma}
\newtheorem{Satz}[Def]{Proposition}
\newtheorem{Prop}[Def]{Proposition}

\newtheorem{Thm}[Def]{Theorem}
\newtheorem{Bsp}[Def]{Example}
\makeatletter
\newenvironment{proof}{\noindent{\textit{Proof:}}}{%
\unskip\nobreak\hfil\penalty50\hskip1em\null\nobreak
$\Box$
\parfillskip=\z@\finalhyphendemerits=0\endgraf\bigskip}

\let\oldendBsp\endBsp
\def\endBsp{\unskip\nobreak\hfil\penalty50\hskip1em\null\nobreak\hfil%
$\blacksquare$\parfillskip=\z@\finalhyphendemerits=0\endgraf\oldendBsp}
\let\oldendBem\endBem
\def\endBem{\unskip\nobreak\hfil\penalty50\hskip1em\null\nobreak\hfil%
$\blacksquare$\parfillskip=\z@\finalhyphendemerits=0\endgraf\oldendBem}
\makeatother
\author[1]{Julia Eisenberg\thanks{corresponding author: \small {\tt jeisenbe@fam.tuwien.ac.at}}\ }
\author[2]{Paul Kr\"uhner\thanks{peisenbe@liverpool.ac.uk}}
\date{}
\affil[1]{TU Wien}
\affil[2]{University of Liverpool}

\normalsize
\title{Suboptimal Control of Dividends under Exponential Utility}

\begin{document}
\maketitle
\begin{abstract}\noindent
We consider an insurance company modelling its surplus process by a Brownian motion with drift. Our target is to maximise the expected exponential utility of discounted dividend payments, given that the dividend rates are bounded by some constant.
\\The utility function destroys the linearity and the time homogeneity of the considered problem. The value function depends not only on the surplus, but also on time. Numerical considerations suggest that the optimal strategy, if it exists, is of a barrier type with a non-linear barrier. In the related article \cite{ghsz}, it has been observed that standard numerical methods break down in certain parameter cases and no close form solution has been found.
\\
For these reasons, we offer a new method allowing to estimate the distance of an arbitrary smooth enough function to the value function. Applying this method, we investigate the goodness of the most obvious suboptimal strategies - payout on the maximal rate, and constant barrier strategies - by measuring the distance of its performance function to the value function.
\vspace{6pt}
\settowidth\labelwidth{{\it 2010 Mathematical Subject Classification: }}%
                \par\noindent {\it 2010 Mathematical Subject Classification: }%
                \rlap{Primary}\phantom{Secondary}
                93E20\newline\null\hskip\labelwidth
                Secondary 91B30, 60H30
\end{abstract}

\section{Introduction}
Dividend payments of companies is one of the most important factors for analytic investors when they have to decide whether they invest into the firm.
Furthermore, dividends serve as a sort of litmus paper, indicating the financial health of the considered company. Indeed, the reputation, and consequently commercial success of a company with a long record of dividend payments would be negatively impacted in the case the company will drop the payments. On the other hand, new companies can additionally strengthen their position by declaring dividends.
For the sake of fairness, it should be mentioned that there are also some serious arguments against dividend payouts, for example for tax reasons it might be advantageous to withhold dividend payments. A discussion of the pros and contras of dividends distribution is beyond the scope of the present manuscript. We refer to surveys on the topic by Avanzi \cite{avanzi} or Albrecher and Thonhauser \cite{AlbThReview}.

Due to its importance, the value of expected discounted dividends has been for a long time, and still remains, one of the most popular risk measures in the actuarial literature.
Lots of papers have been written on maximizing the dividend outcome in various models. 

Gerber \cite{gerb}, B\"uhlmann \cite{buhl}, Azcue and Muler \cite{azmu}, Albrecher and Thonhauser \cite{alth} are just some of the results obtained since de Finetti's path-breaking paper \cite{defin}. 
\\Shreve, Lehoczky and Gaver \cite{Shreve} considered the problem for a general diffusion process, where the drift and the volatility fulfil some special conditions.
Modelling the surplus process via a Brownian motion with drift, was considered by Asmussen and Taksar \cite{astak}, who could find the optimal strategy to be a constant barrier. 

All the papers mentioned above deal with linear dividend payments. Non-linear dividend payments are considered in Hubalek and Schachermayer \cite{hub}. They apply various utility functions to the dividend rates before accumulation. Their result differs a lot from the classical result described in Asmussen and Taksar \cite{astak}.
An interesting question is to consider the expected ``present utility'' of the discounted dividend payments. It means the utility function will be applied on the value of the discounted dividend payments.

Modeling the surplus by a Brownian motion with drift, Grandits et.\ al applied in \cite{ghsz} an exponential utility function to the value of unrestricted discounted dividends. In other words, they considered the expected utility of the present value of dividends and not the expected discounted utility of the dividend rates. In that paper, the existence of the optimal strategy could not be shown.
We will investigate the related problem where the dividend payments are restricted to a finite rate. 
Note that using a non-linear utility function increases the dimension of the problem. Therefore, as for now, tackling the problem via the Hamilton--Jacobi--Bellman approach in order to find an explicit solution seems to be an unsolvable task. Of course, one can prove the value function to be the unique viscosity solution to the corresponding Hamilton--Jacobi--Bellman equation and try to solve the problem numerically. However, if the maximal allowed dividend rate is quite big the standard numerical methods like finite differences and finite elements break down. We discuss this in Section \ref{bsp}.

In this paper, we offer a new approach. Instead of proving the value function to be the unique viscosity solution to the corresponding Hamilton--Jacobi--Bellman equation, we investigate the ``goodness'' of suboptimal strategies.
We simply choose an arbitrary control with an easy-to-calculate return function and compare its performance, or rather an approximation of its performance, against the unknown value function. 
The method is based on sharp occupation bounds which we find by a method developed for sharp density bounds in Ba\~nos and Kr\"uhner \cite{bakr}. This enables us to make an educated guess and to check if our pick is indeed almost as good as the optimal strategy. 
\\This approach drastically differs from procedures usually used for calculation of the value function in two ways. First, unlike most numerical schemes there is no convergence to the value function, i.e.\ one only gets a bound for the performance of one given strategy but no straightforward procedure to get better strategies. Second, our criterion has almost no influence from the dimension of the problem and is consequently directly applicable in higher dimensions. 

The paper is organised as follows. In the next section, we motivate the problem and derive some basic properties of the value function. In Section 3, we consider the case of the maximal constant dividend rate strategy, the properties of the corresponding return function and the goodness of this strategy (a bound for the distance of the return function to the unknown value function). Section 4 investigates the goodness of a constant barrier strategy.
In Section 5, we consider examples illustrating the classical and the new approach. Finally, in the appendix we gather technical proofs and establish occupation bounds.

\section{Motivation}
\noindent
We consider an insurance company whose surplus is modelled as a Brownian motion with drift
\[
X_t=X_0+\mu t+\sigma W_t\quad t\geq 0\;
\]
where $\mu,\sigma, X_0 \in\mathbb R$ on some filtered space $(\Omega,\mathfrak A,(\mathcal F_t)_{t\geq 0})$ where $\mathcal F$ is assumed to be the right-continuous filtration generated by $X$. $X$ is obviously a Markov-process and we denote by $\mathbb P_{(t,x)}$ respectively $\mathbb E_{(t,x)}$ the probability measure respectively expectation conditioned on $X_t=x$, also we use the notation $\mE_x:=\mathbb E_{(0,x)}$. Further, we assume that the company has to pay out dividends, characterized by a dividend rate. Denoting the dividend rate process by $C$, we can write the ex-dividend surplus process as
\[
X^C_t=x+\mu t+\sigma W_t-\int_0^t C_s\md s\;.
\]  
In the present manuscript we allow only allow dividend rate processes $C$ which are progressively measurable and satisfy $0\le C_s\le \xi$ for some maximal rate $\xi>0$ at any time $s\geq 0$. We call these strategies {\em admissible}. Let $U(x)=\frac1\gamma-\frac1\gamma e^{-\gamma x}$, $\gamma>0$, be the underlying utility function and $\tau^C:=\inf\{s\ge t:X_s^C<0\}$ the ruin time corresponding to the strategy $c$ under the measure $\mathbb P_{(t,x)}$. Our objective is to maximize the expected exponential utility of the discounted dividend payments until ruin. Since we can start our observation in every time point $t\in\R_+$, the target functional is given by
\[
V^C(t,x):=\mE_{(t,x)}\Big[U\Big(\int_t^{\tau^C} e^{-\delta s}C_s\md s\Big)\Big]\;.
\]
Here, we assume that the dividend payout up to $t$ equals $0$, for a rigorous simplification confer \cite{ghsz} or simply note that with already paid dividends $\bar C$ up to time $t$ we have

$$ \mE_{(t,x)}\Big[U\Big(\bar C+ \int_t^{\tau^C} e^{-\delta s}C_s\md s\Big)|\mathcal F_t\Big] = U(\bar C)+e^{-\gamma \bar C}V^C(t,x). $$
The corresponding value function $V$ is defined by
\[
V(t,x):=\sup\limits_{C}\mE_{(t,x)}\Big[U\Big(\int_t^{\tau^C} e^{-\delta s}C_s\md s\Big)\Big]\;
\]
where the supremum is taken over all admissible strategies $C$. Note that $V(t,0)=0$, because ruin will happen immediately due to the oscillation of Brownian motion, i.e.\ $\tau^C = \min\{s\geq t: X_s^C=0\}$ for any strategy $C$ under $\PP$.
The Hamilton--Jacobi--Bellman (HJB) equation corresponding to the problem can be found similar as in \cite{ghsz}, for general explanations confer for instance \cite{HS}:
\begin{equation}
V_t+\mu V_x+\frac{\sigma^2}{2}V_{xx}+\sup\limits_{0\le y\le\xi}\Big[y\Big(-V_x+e^{-\delta t}(1-\gamma V)\Big)\Big]=0\;.\label{hjb}
\end{equation}
We like to stress at this point that we neither show that the value function solves the (HJB) in some sense, nor that a good enough solution is the value function. In fact, our approach of evaluating the goodness of a given strategy compared to the unknown optimal strategy does not assume any knowledge about the optimal strategy.

Assuming that the optimal strategy $C^*$ exists and that the value function does satisfy the HJB in a suitable sense, we have for the optimal dividend rates
\[
C^*_s=
\begin{cases}
0 &{\rm if} \;\; -V_x+e^{-\delta s}(1-\gamma V)<0\\
\in[0,\xi]& {\rm if}\;\; -V_x+e^{-\delta s}(1-\gamma V)=0\\
\xi & {\rm if}\;\; -V_x+e^{-\delta s}(1-\gamma V)>0\;.
\end{cases}
\]
$\PP$-a.s.\ for any $s\geq t$.
\begin{Bem} \label{uniform}
For every dividend strategy $C$ it holds: 
\begin{equation*}
V^C(t,x)= \mE_{(t,x)}\Big[U\Big(\int_t^{\tau^C}C_s e^{-\delta s}\md s \Big)\Big]\le U\Big(\xi\int_t^{\infty} e^{-\delta s}\md s \Big)= U\Big(\frac\xi\delta e^{-\delta t} \Big)
\end{equation*}
We conclude
\[
\lim\limits_{x\to\infty}V(t,x)\le U\Big(\frac\xi\delta e^{-\delta t}\Big)\;,
\]
and $V$ is a bounded function. Consider now a constant strategy $C_t\equiv\xi$, i.e. we always pay on the rate $\xi$. The ex-dividend process becomes a Brownian motion with drift $\mu-\xi$ and volatility $\sigma$.
Define further
\begin{equation}
\eta_n:=\frac{\xi-\mu-\sqrt{(\xi-\mu)^2+2\delta\sigma^2 n}}{\sigma^2}\;,\label{eta}
\end{equation}
and let $\tau^\xi:=\inf\{s\ge t:\; x+(\mu-\xi) s+\sigma W_s\le 0\}$, i.e. $\tau^\xi$ is the ruin time under the strategy $\xi$. Here and in the following we define
\begin{equation} 
\de:= \xi\gamma/\delta. \label{de}
\end{equation}
With help of change of measure technique, see for example \cite[p. 216]{HS}, we can calculate the return function $V^\xi$ of the constant strategy $C_t\equiv\xi$ by using the power series representation of the exponential function:
\begin{align}
V^\xi(t,x)&= \mE_{x}\Big[U\Big(\xi\int_t^{\tau^\xi}e^{-\delta s}\md s\Big)\Big]
=\frac 1\gamma-\frac 1\gamma\mE_{x}\Big[e^{-\de \big(e^{-\delta t}-e^{-\delta(t+\tau^\xi)}\big)}\Big]\nonumber
\\&=\frac 1\gamma-\frac 1\gamma e^{-\de e^{-\delta t}}\mE_{x}\Big[e^{\de e^{-\delta(t+\tau^\xi)}}\Big]=\frac 1\gamma -\frac { e^{-\de e^{-\delta t}}}\gamma\s_{n=0}^\infty \frac{e^{-\delta tn}\de^n}{n!}\mE_x[e^{-\delta \tau^\xi n}]\nonumber
\\&=\frac 1\gamma -\frac { e^{-\de e^{-\delta t}}}\gamma-\frac { e^{-\de e^{-\delta t}}}\gamma\s_{n=1}^\infty \frac{e^{-\delta tn}\de^n}{n!}e^{\eta_n x}\;.\label{vxi}
\end{align}
It is obvious, that in the above power series $\lim\limits_{x\to\infty}$ and summation can be interchanged yielding $\lim\limits_{x\to\infty} V^\xi(t,x)=U\Big(\frac\xi\delta e^{-\delta t}\Big)$.
In particular, we can now conclude% using Remark \ref{uniform}
\[
\lim\limits_{x\to\infty}V(t,x)=\frac 1\gamma-\frac 1\gamma\exp\big(-\Delta e^{-\delta t}\big)=U\Big(\frac\xi\delta e^{-\delta t}\Big)\;.
\]
uniformly in $t\in[0,\infty)$.
\end{Bem}
Next, we show that for some special values of the maximal rate $\xi$ with a positive probability the ex-dividend surplus process remains positive up to infinity.
\begin{Bem}
Let $C$ be an admissible strategy,
where $X^C$ is the process under the strategy $C$. Let further $X^\xi$ be the process under the constant strategy $\xi$, i.e. $X^\xi$ is a Brownian motion with drift $(\mu-\xi)$ and volatility $\sigma$.
Then it is clear that
\[
X_s^\xi\le X^C_s\;.
\]
If $\mu>\xi$ then it holds, see for example \cite[p. 223]{bs}, $\mP_{(t,x)}[\tau^C=\infty]\ge\mP_{(t,x)}[\tau^\xi= \infty]>0$.  
\end{Bem}

Finally, we gather one structural property of the value function which, however, is not used later.
\begin{Satz}
The value function is Lipschitz-continuous, strictly increasing in $x$ and decreasing in $t$.
\end{Satz}
\begin{proof}
It is clear that $V$ is strictly increasing in $x$.
\\Consider further $(t,0)$ with $t\in\R_+$. Let $h,\varepsilon>0$ and $C$ be an arbitrary admissible strategy. Let $\tau^0$ be the ruin time for the strategy which is constant zero. 
Define
\begin{equation}
\varrho_n:=\frac{\sqrt{\mu^2+2\sigma^2 \delta n}}{\sigma^2}\;,\quad 
\theta_n:=\frac{-\mu}{\sigma^2}+\varrho_n \quad \mbox{and}\quad\zeta_n:=\frac{-\mu}{\sigma^2}-\varrho_n \;\label{tz}
\end{equation}
for any $n\in\mathbb N$. Using $\mE_{h}[e^{-\delta \tau^0}]=e^{\zeta_1 h}$, confer for instance \cite[p.\ 295]{bs}, it follows with $X^0_s\ge X_s^C$:
\begin{align}
V^C(t,h)&=\mE_{(t,h)}\bigg[U\Big(\int_{t}^{\tau^C}e^{-\delta s}C_{s}\md s\Big)\bigg]
 \le \mE_{h}\bigg[U\Big(\xi\int_{t}^{t+\tau^0}e^{-\delta s}\md s\Big)\bigg]\nonumber
 \\&=\mE_{h}\bigg[U\Big(\frac\xi\delta e^{-\delta t}(1-e^{-\delta \tau^0})\Big)\bigg]
\le \frac\xi\delta  e^{-\delta t}\big(1-e^{\zeta_1 h}\big)\le -\frac\xi\delta \zeta_1 h \label{motiv:est}\;.
\end{align} 
\\-- Let $h\geq 0$ and $\tau^0$ be the ruin time for the strategy which is constant zero. 
Let $(t,x)\in\R_+^2$ be arbitrary, $C$ be an admissible strategy which is $\varepsilon$-optimal for the starting point $(t,x+h)$, i.e. $V(t,x+h)-V^C(t,x+h)\le \varepsilon$.
Define further $\tilde\tau:=\inf\{s\ge t: \; X^C_s=h\}$. Then $X_s\geq 0$ for $s\in [t,\tilde\tau]$ under $\PP$ because $X_s\geq h$ for $s\in[t,\tilde\tau]$ under $\mP_{(t,x+h)}$.
Since $U$ fulfils $U(a+b)\le U(a)+U(b)$ for any $a,b\geq 0$, we have
\begin{align*}
V(t,x+h)&\le V^{C}(t,x+h)+\varepsilon
=\mE_{(t,x+h)}\bigg [U\bigg(\int_{t}^{\tau^C} e^{-\delta s} C_{s} \md s\bigg)\bigg]+\varepsilon
\\&=\mE_{(t,x+h)}\bigg[U\bigg(\int_{t}^{\tilde \tau} e^{-\delta s} C_{s} \md s+\int_{\tilde\tau}^{\tau^C} e^{-\delta s} C_{s} \md s\bigg)\bigg]+\varepsilon
%\\&\ge \mE\bigg[\one_{\{\tau^\xi>y\}}U\bigg(\frac{\xi}\delta e^{-\delta t}\big(1-e^{-\delta y}\big)+e^{-\delta y}\int_{t}^{\tau^C} e^{-\delta s} C_{s} \md s\bigg)\bigg]+\varepsilon
\\&\le \mE_{(t,x+h)}\bigg[U\bigg(\int_{t}^{\tilde \tau} e^{-\delta s} C_{s} \md s\bigg)\bigg]+\mE_{(t,x+h)}\bigg[U\bigg(\int_{\tilde\tau}^{\tau^C} e^{-\delta s} C_{s} \md s\bigg)\bigg]+\varepsilon
\\&\le V^C(t,x)+\mE_{(t,x+h)}\bigg[U\bigg(\frac\xi\delta\big(e^{-\delta \tilde \tau}-e^{-\delta\tau^C}\big)\bigg)\bigg]+\varepsilon
\\&\le V(t,x)+ \mE_{h}\bigg[U\bigg(\frac\xi\delta\big(1-e^{-\delta\tau^0}\big)\bigg)\bigg]+\varepsilon
\;.
\end{align*}
The last inequality follows from monotony of $U$ and the fact that the $\mP_{(\tilde\tau,h)}$-random time $\tau^C-\tilde\tau$ is less or equal than $\tau^0$ $\mP_{(\tilde\tau,h)}$-a.s. Because $\varepsilon$ was arbitrary and due to \eqref{motiv:est} we find
\[
0\le V(t,x+h)-V(t,x)\le -\frac\xi\delta \zeta_1 h\;.
\]
Consequently, $V$ is Lipschitz-continuous in the space variable $x$ with Lipschitz-constant at most $-\frac\xi\delta \zeta_1$.
\medskip
\\-- Next, we consider the properties of the value function concerning the time variable. Because $\delta>0$, it is clear that $V$ is strictly decreasing in $t$. 
\\
Let further $(t,x)\in\R_+^2$, $h>0$ and $C$ be an admissible strategy. Then, the strategy $\tilde C$ with $\tilde C_s:=C_{s-h}\one_{\{s\geq h\}}$ is admissible. Since, $U$ is concave we have
\begin{align*}
V(t+h,x)&\ge V^{\tilde C}(t+h,x)= \mE_{(t+h,x)}\bigg[U\Big(\int_{t+h}^{\tau^C+h}e^{-\delta s}C_{s-h}\md s\Big)\bigg]
\\&=\mE_{(t,x)}\bigg[U\Big(e^{-\delta h}\int_{t}^{\tau^C}e^{-\delta s}C_{s}\md s\Big)\bigg]
\ge e^{-\delta h}V^C(t,x)\;.
\end{align*}
Building the supremum over all admissible strategies on the right side of the above inequality and using Remark \ref{uniform}, yields 
\[
0\ge V(t+h,x)-V(t,x)\ge V(t,x)(e^{-\delta h}-1)\ge - U\Big(\frac\xi\delta\Big)\delta h
\]
and, consequently, $V$ is Lipschitz-continuous as a function of $t$ with constant $\delta U(\xi/\delta)$.
\end{proof}

\section{Payout on the Maximal Rate \label{maxrate}}
\subsection{Could it be optimal to pay on the maximal rate up to ruin?}
\noindent
At first, we investigate the constant strategy $\xi$, i.e.\ the dividends will be paid out at the maximal rate $\xi$ until ruin. In this section we find exact conditions under which this strategy is optimal. We already know from \eqref{vxi} that the corresponding return function is given by 
\[
V^\xi(t,x)=\frac1\gamma-\frac1\gamma e^{-\de e^{-\delta t}}-e^{-\de e^{-\delta t}}\s_{n=1}^\infty \frac{\de^n}{\gamma n!}e^{-\delta tn}e^{\eta_n x}\;.
\]
It is obvious that $V^\xi$ is increasing and concave in $x$ and decreasing in $t$. For further considerations we will need the following remark. 
\begin{Bem}\label{bem:prop}
Consider $\eta_n$, defined in \eqref{eta}, as a function of $\xi$.  \medskip
\\$\mathbf{1.}$ Since $$\frac\md{\md\xi}\eta_n=\frac{-\eta_n}{\sqrt{(\xi-\mu)^2+2\delta\sigma^2 n}},$$ it is easy to see that $\eta_n(\xi)$ and $\frac{\eta_{n+1}(\xi)n}{\eta_n(\xi)(n+1)}$ are increasing in $\xi$. Also, we have 
\[
\lim\limits_{\xi\to\infty}\frac{\eta_{n+1}(\xi)n}{\eta_n(\xi)(n+1)}=1\;.
\]
We conclude that $\frac{\eta_{n+1}}{(n+1)}>\frac{\eta_{n}}{n}$.\medskip 
\\$\mathbf{2.}$ Further, we put to record
\[
\lim\limits_{\xi\to\infty}\xi\eta_n(\xi)=-\delta n\;.
\]
$\mathbf{3.}$ Also, we have
\begin{align*}
\frac{\md}{\md \xi}\Big(\delta n+\xi\eta_n(\xi)\Big)&
=\eta_n\Big(1-\frac{\xi}{\sqrt{(\xi-\mu)^2+2\delta\sigma^2n}}\Big)
\begin{cases}
<0 & \xi<\frac{\mu^2+2\delta\sigma^2n}{2\mu}
\\\ge 0 &\xi\ge\frac{\mu^2+2\delta\sigma^2n}{2\mu}\;.
\end{cases}
\end{align*}
Thus, at $\xi=0$ the function $\xi\mapsto \delta n+\xi\eta_n(\xi)$ attains the value $\delta n>0$, at its minimum point $\xi^*=\frac{\mu^2+2\delta\sigma^2n}{2\mu}$ we have
\[
\delta n+\xi^*\eta_n(\xi^*)=\delta n-\frac{\mu^2+2\delta\sigma^2n}{2\sigma^2}=-\frac{\mu^2}{2\sigma^2}<0\;
\]
and, finally, for $\xi\to\infty$ it holds, due to Item 2 above, that $\lim\limits_{\xi\to 0}\delta n+\xi\eta_n(\xi)=0$. Thus, for every $n\in\N$ the function $\xi \mapsto 1+\frac{\eta_n(\xi)\xi}{\delta n}$ has a unique zero at $\frac{\delta n\sigma^2}{2\mu}$.
\end{Bem}
Further, it is easy to check that in $V^\xi$ summation and differentiation can be interchanged.
Derivation with respect to $x$ yields
\[
V^\xi_x(t,x)=-e^{-\de e^{-\delta t}}\s_{n=1}^{\infty}\frac{\de^n}{\gamma n!}e^{-\delta tn}\eta_n e^{\eta_nx}\;.
\]
In order to answer the optimality question, we have to look at the function $-V_x^\xi+e^{-\delta t}\big(1-\gamma V^\xi\big)$.
For simplicity, we multiply the above expression by $e^{\delta t}e^{\de e^{-\delta t}}$, substitute $e^{-\delta t}$ by $t\in[0,1]$ and define
\begin{align}
\psi(t,x):&=
%\frac{e^{\de t}}t\Big\{-h_x(t,x)+t\big(1-\gamma h(t,x)\big)\Big\}
\s_{n=0}^{\infty}t^{n}\frac{\de^n}{n!}\Big\{\frac{\eta_{n+1}\xi}{\delta (n+1)}e^{\eta_{n+1}x}+e^{\eta_n x}\Big\}\;.\label{psi}
%+\big(\eta_1\frac{\xi}{\delta}e^{\eta_{1}x}+1\big)\;.
%\\&=\s_{n=0}^{\infty}t^n\eta_{n+1}(\xi)\frac{\xi^n\gamma^{n-1}}{\delta^nn!}e^{\eta_n(\xi)x}+\s_{n=1}^\infty \frac{\xi^n\gamma^{n}}{\delta^nn!}t^{n}e^{\eta_n(\xi)x}+1\;.
\end{align}
If $\psi\ge0$ on $[0,1]\times \R_+$, then $V^\xi$ does solve the HJB equation and as we will see, it is the value function in that case. 
\begin{Satz}\label{p:0 opt}
$V^\xi$ is the value function if and only if $\xi\le \frac{\delta\sigma^2}{2\mu}$. In that case $V^\xi$ is a classical solution to the HJB equation \eqref{hjb} and an optimal strategy is constant $\xi$.
\end{Satz}
\begin{proof}
It is easy to check that $V^\xi$ solves the differential equation 
\[
V_t+\mu V_x+\frac{\sigma^2}{2}V_{xx}+\xi\Big(-V_x+e^{-\delta t}(1-\gamma V)\Big)=0\;.
\]

We first assume that $\xi\leq \frac{\delta\sigma^2}{2\mu}$ and show that $V^\xi$ is the value function. Then we have $\xi \leq n\frac{\delta\sigma^2}{2\mu}$ for any $n\geq 1$ and, hence Remark \ref{bem:prop} yields
 $$ \eta_n\frac{\xi}{\delta n}+1>0. $$ 
 This yields immediately for all $(t,x)\in (0,1]\times \R_+$:
\begin{align*}
\psi(t,x)\ge \s_{n=0}^{\infty}t^{n}\frac{\de^n}{n!}\Big\{\frac{\eta_{n+1}\xi}{\delta(n+1)}+1\Big\}e^{\eta_n x} \ge 0\;.
%+\big(\eta_1\frac{\xi}{\delta}e^{\eta_{1}x}+1\big)\ge 0\;.
\end{align*}
Let now $C$ be an arbitrary admissible strategy, $\tau$ its ruin time and $\hat X_u:=X^C_u$. Applying Ito's formula yields $\PP$-a.s.

\begin{align*}
&e^{-\gamma\int_t^{\tau\w s} e^{-\delta u}C_u\md u}V^\xi(\tau\w s,\hat X_{\tau\w s})
=V^\xi(t,x)+\sigma\int_t^{\tau\w s} e^{-\gamma\int_t^{y} e^{-\delta u}C_u\md u}V^\xi_x\md W_y
\\&{}\quad+\int_t^{\tau\w s} e^{-\gamma\int_t^{y} e^{-\delta u}C_u\md u}\Big\{V^\xi_t+(\mu-C_y)V^\xi_x+\frac{\sigma^2}2V^\xi_{xx}-\gamma C_y e^{-\delta y} V^\xi\Big\}\md y\;.
\end{align*}
Since $V_x^\xi$ is bounded, the stochastic integral above is a martingale with expectation zero. For the second integral one obtains using the differential equation for $V^\xi$:
\begin{align*}
\int_t^{\tau\w s} e^{-\gamma\int_t^{y} e^{-\delta u}C_u\md u}\Big\{V^\xi_t+(\mu-C_y)V^\xi_x+\frac{\sigma^2}2V^\xi_{xx}-\gamma C_y e^{-\delta y} V^\xi\Big\}\md y
\\= \int_t^{\tau\w s} e^{-\gamma\int_t^{y} e^{-\delta u}C_u\md u}\Big\{\big(C_y-\xi\big)\psi\big(e^{-\delta y},\hat X_y\big)-C_y e^{-\delta y}\Big\}\md y\;.
\end{align*}
Building the expectations on the both sides and letting $s\to\infty$, we obtain by interchanging limit and expectation (due to the bounded convergence theorem):
\begin{align*}
0&=V^\xi(t,x)+\mE_{(t,x)}\Big[\int_t^{\tau} e^{-\gamma\int_t^{y} e^{-\delta u}C_u\md u}\Big\{\big(C_y-\xi\big)\psi\big(e^{-\delta y},\hat X_y\big)-C_y e^{-\delta y}\Big\}\md y\Big]
\\&=V^\xi(t,x)+ \mE_{(t,x)}\Big[\int_t^{\tau} e^{-\gamma\int_t^{y} e^{-\delta u}C_u\md u}\big(C_y-\xi\big)\psi\big(e^{-\delta y},\hat X_y\big)\md y\Big]-V^C(t,x)\;.
\end{align*}
Since $C_u\le \xi$ and $\psi\ge 0$, the expectation above is non-positive, giving $V^C(t,x)\le V^\xi(t,x)$ for all admissible strategies $C$. Therefore, $V^\xi$ is the value function. 
\\Assume now $\xi> \frac{\delta\sigma^2}{2\mu}$ and we assume for contradiction that $V^\xi$ is the value function. Then we have $\psi(0,0)=1+\eta_1 \xi/\delta<0$. It means in particular, that the function $\psi$ is negative also for some $(t,x)\in(0,1]\times\R_+$. Consequently, $V^\xi$ does not solve the HJB equation \eqref{hjb}. Moreover, $V^\xi$ is smooth enough and has bounded $x$-derivative. Thus, classical verifaction results, cf.\ \cite{flso}, yield that $V^\xi$ solves the HJB equation. A contradiction.
\end{proof}
In the following, we assume $\xi>\frac{\delta\sigma^2}{2\mu}$.

\subsection{The goodness of the strategy $\xi$.}\label{s:xi good}
We now provide an estimate on the goodness of the constant payout strategy which relies only on the performance of the chosen strategy $\xi$ and on deterministic constants. Recall from \eqref{eta} and \eqref{tz} that
\begin{align*}
&\eta_n = \frac{(\xi-\mu)-\sqrt{(\xi-\mu)^2+2n\delta\sigma^2}}{\sigma^2}, \\ &\theta_n = \frac{-\mu+\sqrt{\mu^2+2n\delta\sigma^2}}{\sigma^2}, 
\quad \zeta_n = \frac{-\mu-\sqrt{\mu^2+2n\delta\sigma^2}}{\sigma^2}.
   \end{align*}
\begin{Satz}\label{p:Goodness constant}
 Let $t,x\geq 0$. Then we have
\begin{align*}
V(t,x) &\le V^\xi(t,x) 
\\&\quad {}+\xi e^{-\delta t-\Delta e^{-\delta t}}\s_{n=0}^{\infty} e^{-\delta t n} \frac{\de^n}{n!} \int_0^{\infty} \left(\frac{-\eta_{n+1}\xi}{\delta(n+1)}e^{\eta_{n+1}y}-e^{\eta_{n}y}\right)^+f_{n+1}(x,y)\md y,
\end{align*}
 where
     \begin{align*}
      f_n(x,y) &:= \frac{2\left(e^{\theta_n(x \wedge y)}-e^{\zeta_n(x\wedge y)}\right)e^{\eta_{n}(x-y)^+}}{\sigma^2\left((\theta_n-\eta_n)e^{y\theta_n}-(\zeta_n-\eta_n)e^{y\zeta_n}\right)},\quad y\geq 0.
  \end{align*}
\end{Satz}
\begin{proof}
We know that the return function $V^\xi\in \mathcal C^{1,2}$. Let $C$ be an arbitrary admissible strategy. Then, using Ito's formula for $s>t$ under $\PP$:
\begin{align*}
&e^{-\gamma\int_t^{s\w\tau^C} e^{-\delta u}C_u\md u}V^\xi(s\w\tau^C,X_{s\w\tau^C}^C)
\\&=V^\xi(t,x)+\int_t^{s\w\tau^C} e^{-\gamma\int_t^{r} e^{-\delta u}C_u\md u}\Big\{V^\xi_t+(\mu-C_r) V^\xi_x+\frac{\sigma^2}2 V^\xi_{xx}-\gamma e^{-\delta r}C_rV^\xi\Big\}\md r
\\&\quad {}+\sigma\int_t^{s\w\tau^C} e^{-\gamma\int_t^{r} e^{-\delta u}C_u\md u} V^\xi_x\md W_r\;.
\end{align*}
Using the differential equation for $V^\xi$, one obtains as in the last proof:
\begin{align*}
e^{-\gamma\int_t^{s\w\tau^C} e^{-\delta u}C_u\md u}&V^\xi(s\w\tau^C,X_{s\w\tau^C}^C)
%\\&= V^\xi(t,x)+\xi\int_t^{s\w\tau^C} e^{-\gamma\int_t^{r} e^{-\delta u}C_u\md u}\Big\{V^\xi_x-e^{-\delta r}\big(1-\gamma  V^\xi\big)\Big\}\md r
%\\&\quad {}-\int_t^{s\w\tau^C} e^{-\gamma\int_t^{r} e^{-\delta u}C_u\md u}C_r\Big\{V^\xi_x+\gamma e^{-\delta r} V^\xi\Big\}\md r+\sigma\int_t^{s\w\tau^C} e^{-\gamma\int_t^{r} C_u\md u} V^\xi_x\md W_r
\\&= V^\xi(t,x)+\int_t^{s\w\tau^C} e^{-\gamma\int_t^{r} e^{-\delta u}C_u\md u}(C_r-\xi)\psi(e^{-\delta r},X^C_r)\md r
\\&\quad {}-\int_t^{s\w\tau^C} e^{-\gamma\int_t^{r} C_u\md u}C_r e^{-\delta r}\md r+\sigma\int_t^{s\w\tau^C} e^{-\gamma\int_t^{r} C_u\md u} V^\xi_x\md W_r
\end{align*}
Building the $\PP$-expectations, letting $s\to\infty$ and rearranging the terms, one has
\begin{align*}
V^C(t,x)=V^\xi(t,x)+\mE_{(t,x)}\Big[\int_t^{\tau^C} e^{-\gamma\int_t^{r} C_u\md u}(C_r-\xi)\psi(e^{-\delta r},X^C_r)\md r\Big]\;.
\end{align*}
Our goal is to find a $C$-independent estimate for the expectation on the rhs.\ above, in order to gain a bound for the difference $V(t,x)-V^\xi(t,x)$. We have
\begin{align*}
&\mE_{(t,x)}\Big[\int_t^{\tau^C} e^{-\gamma\int_t^{r} C_u\md u}(C_r-\xi)\psi(e^{-\delta r},X^C_r)\md r\Big]
\\ &\le -\xi\mE_{(t,x)}\Big[\int_t^{\tau^C} e^{-\gamma\int_t^{r} C_u\md u}\psi(e^{-\delta r},X^C_r)\one_{\{\psi(e^{-\delta r},X^C_r)<0\}}\md r\Big]\; 
\\
 &\le -\xi\mE_{(t,x)}\Big[\int_t^{\tau^C} \exp(-\delta r)\exp(-\Delta e^{-\delta r})\psi(e^{-\delta r},X^C_r)\one_{\{\psi(e^{-\delta r},X^C_r)<0\}}\md r\Big]\; \\
 &\le -\xi e^{-\Delta e^{-\delta t}}\mE_{(t,x)}\Big[\int_t^{\tau^C} \exp(-\delta r)\psi(e^{-\delta r},X^C_r)\one_{\{\psi(e^{-\delta r},X^C_r)<0\}}\md r\Big]\; \\
 & \le \xi e^{-\Delta e^{-\delta t}}\s_{n=0}^{\infty} e^{-\delta t (n+1)} \frac{\de^n}{n!}\mE_{(t,x)}\Big[ \int_t^{\tau^C} e^{-\delta (r-t)(n+1)}\left(\frac{-\eta_{n+1}\xi}{\delta(n+1)}e^{\eta_{n+1}X_r^C}-e^{\eta_{n}X_r^C}\right)^+{\rm d} r\Big]\; \\
 & \le \xi e^{-\Delta e^{-\delta t}-\delta t}\s_{n=0}^{\infty} e^{-\delta t n} \frac{\de^n}{n!} \int_0^{\infty} \left(\frac{-\eta_{n+1}\xi}{\delta(n+1)}e^{y\eta_{n+1}}-e^{\eta_{n}y}\right)^+f_{n+1}(x,y)\md y\; 
\end{align*}
where the last inequality follows from Theorem \ref{t:occupation bound}.
\end{proof}

\section{Goodness of Constant Barrier Strategies}
\noindent
Shreve et al.\ \cite{Shreve} and Asmussen and Taksar \cite{astak} considered the problem of dividend maximization for a surplus described by a Bownian motion with drift. The optimal strategy there turned out to be a barrier strategy with a constant barrier. 
\\Let $q\in\R_+$ and $C$ be given by $C_s=\xi\one_{\{X^C_s>q\}}$, i.e.\ $C$ is a barrier strategy with a constant barrier $q$ and ruin time $\tau^C=\inf\{s\ge 0:\;X^C_s=0\}$. The corresponding return function fulfils due to the Markov-property of $X^C$
\[
V^C(t,x)=\frac1\gamma -\frac 1\gamma \mE_{x}\Big[e^{-\gamma \int_t^{t+\tau^C}e^{-\delta s} C_s\md s}\Big]\;.
\] 
Note that for every $a>0$ we have
\begin{align*}
\mE_{x}\Big[e^{a \int_t^{t+\tau^C}e^{-\delta s} C_s\md s}\Big]&\le e^{a \int_t^{\infty}e^{-\delta s} \xi\md s}=e^{\frac{a\xi}\delta e^{-\delta t}}<\infty\;.
\end{align*}
It means, the moment generating function of $\int_t^{t+\tau^C}e^{-\delta s} C_s\md s$ is infinitely often differentiable and all moments of $\int_t^{t+\tau^C}e^{-\delta s} C_s\md s$ exist. We define  
\begin{align*}
%&\tau_q^C:=\inf\{s\ge 0:\;X^C_s=0,\; X_0=q\},
&M_n(q):=\E_q\Big[\Big(\Delta-\gamma \int_0^{\tau^C}e^{-\delta s} C_s\md s\Big)^n\Big]>0,
%&Z_n(q):=\E\Big[\Big(1-\delta \int_0^{\tau_q^C}e^{-\delta s} 1_{\{X_s>q\}}\md s\Big)^n\Big]\in(0,1),
%\\&\tilde M_n(q):=(-1)^n\mE\Big[\Big(\gamma \int_0^{\tau_q^C}e^{-\delta s} C_s\md s\Big)^n\Big]
\\&\tau^{q,\xi}:= \inf\{s\ge 0:\;X_s^\xi=q\},
\\&\tau^{q,0}:= \inf\{s\ge 0:\;X_s^0\notin(0,q)\}.
\end{align*}
Then, for $F(t,x):=V^C(t,x)$, $x>q$, and for $G(t,x):=V^C(t,x)$, $x<q$, it holds:
\begin{align}
F(t,x&)=\frac1\gamma -\frac 1\gamma \mE_{x}\Big[e^{-\gamma\xi\int_t^{t+\tau^{q,\xi}} e^{-\delta s}\md s-\gamma \int_{t+\tau^{q,\xi}}^{t+\tau^C}e^{-\delta s} C_s\md s}\Big]\nonumber
\\ &= \frac1\gamma -\frac 1\gamma \mE_{x}\Big[\exp\big(e^{-\delta t}(-\Delta (1-e^{-\delta \tau^{q,\xi}})-\gamma \int_{\tau^{q,\xi}}^{\tau^C}e^{-\delta s} C_s\md s)\big)\Big]\nonumber
\\ &= \frac1\gamma -\frac 1\gamma e^{-\Delta e^{-\delta t}}\mE_{x}\Big[\exp\big(e^{-\delta t}e^{-\delta \tau^{q,\xi}}(\Delta-\gamma \int_{0}^{\tau^C-\tau^{q,\xi}}e^{-\delta s} C_s\md s)\big)\Big]\nonumber
%\\F(t,x)&=\frac1\gamma -\frac 1\gamma \mE_{x}\Big[e^{-\gamma\xi\int_t^{t+\tau^{q,\xi}} e^{-\delta s}\md s-\gamma \int_{t+\tau^{q,\xi}}^{t+\tau^C}e^{-\delta s} C_s\md s}\Big]\nonumber
%\\&=\frac 1\gamma -\frac 1\gamma \mE_{x}\Big[e^{-\Delta e^{-\delta t}(1-e^{-\delta \tau^{q,\xi}})-\gamma e^{-\delta(t+\tau^{q,\xi})} \int_{0}^{\tau_q^C}e^{-\delta s} C_s\md s}\Big]\nonumber
\\&=\frac 1\gamma-\frac 1\gamma e^{-\Delta e^{-\delta t}}-\frac 1\gamma e^{-\Delta e^{-\delta t}}\s_{n=1}^\infty \frac{e^{-\delta tn}}{n!}\mE_{x}\big[e^{-\delta n\tau^{q,\xi}}\big]\mE_q\Big[\Big(\Delta-\gamma\int_0^{\tau^C}e^{-\delta s} C_s\md s\Big)^n\Big]\nonumber
\\&= \frac 1\gamma-\frac 1\gamma e^{-\Delta e^{-\delta t}}-\frac 1\gamma e^{-\Delta e^{-\delta t}}\s_{n=1}^\infty \frac{e^{-\delta tn}}{n!}e^{\eta_n(x-q)} M_n(q)\nonumber
\\&= -\frac 1\gamma \s_{n=1}^\infty \frac{e^{-\delta tn}}{n!}\s_{k=0}^n\binom{n}{k}(-\Delta)^{n-k} M_k(q)e^{\eta_k(x-q)}\;.\label{F2}
\\\vspace{0.5cm}\nonumber
\\
G(t,x&)=\mE_{x}\big[F\big(t+\tau^{q,0},q\big);X^0_{\tau^{q,0}}=q\big]\nonumber
%\frac1\gamma -\frac 1\gamma \mE\Big[e^{-\gamma \int_{t+\tau^{q,0}}^{t+\tau_x^{q,0}+\tau_q^C}e^{-\delta s} c_s\md s};X^0_{\tau_x^{q,0}}=q\Big]\nonumber
%\\&=\frac 1\gamma -\frac 1\gamma \mE\Big[e^{-\gamma e^{-\delta(t+\tau_x^{q,0})} \int_{0}^{\tau_q^C}e^{-\delta s} c_s\md s};X^0_{\tau_x^{q,0}}=q\Big]\nonumber
%\\&=-\frac 1\gamma\s_{n=1}^\infty \frac{e^{-\delta tn}}{n!}\mE\big[e^{-\delta n \tau_x^{q,0}};X^0_{\tau_x^{q,0}}=q\big]\mE\Big[\Big(-\gamma\int_0^{\tau_q^C}e^{-\delta s} c_s\md s\Big)^n\Big]\nonumber
%\\&=-\frac 1\gamma\s_{n=1}^\infty \frac{e^{-\delta tn}}{n!}\cdot e^{-\frac\mu{\sigma^2} (x-q)}\frac{\sinh(\varrho_n x)}{\sinh(\varrho_n q)}\tilde M_n(q)\label{G2}
\\&=-\frac 1\gamma\s_{n=1}^\infty \frac{e^{-\delta tn}}{n!}\cdot\frac{e^{\theta_nx}-e^{\zeta_n x}}{e^{\theta_nq}-e^{\zeta_n q}}\s_{k=0}^n\binom{n}{k}(-\Delta)^{n-k}M_k(q)\;. \label{G1} 
%\&=-\frac 1\gamma\s_{n=1}^\infty \frac{e^{-\delta tn}}{n!}\cdot\frac{e^{\theta_nx}-e^{\zeta_n x}}{e^{\theta_nq}-e^{\zeta_n q}}(-\Delta)^nZ_n(q)\;. 
\end{align}
where, for the fourth equality, we developed the first exponential function in the expectation into its power series and used the Markov property to see that the $\mP_{0,x}$-law given $\mathcal F_{\tau^{q,\xi}}$ of $\tau^C-\tau^{q,\xi}$ equals the $\mP_{0,q}$-law of $\tau^C$.

\begin{figure}[t]
\includegraphics[scale=0.6, bb = -40 0 200 200]{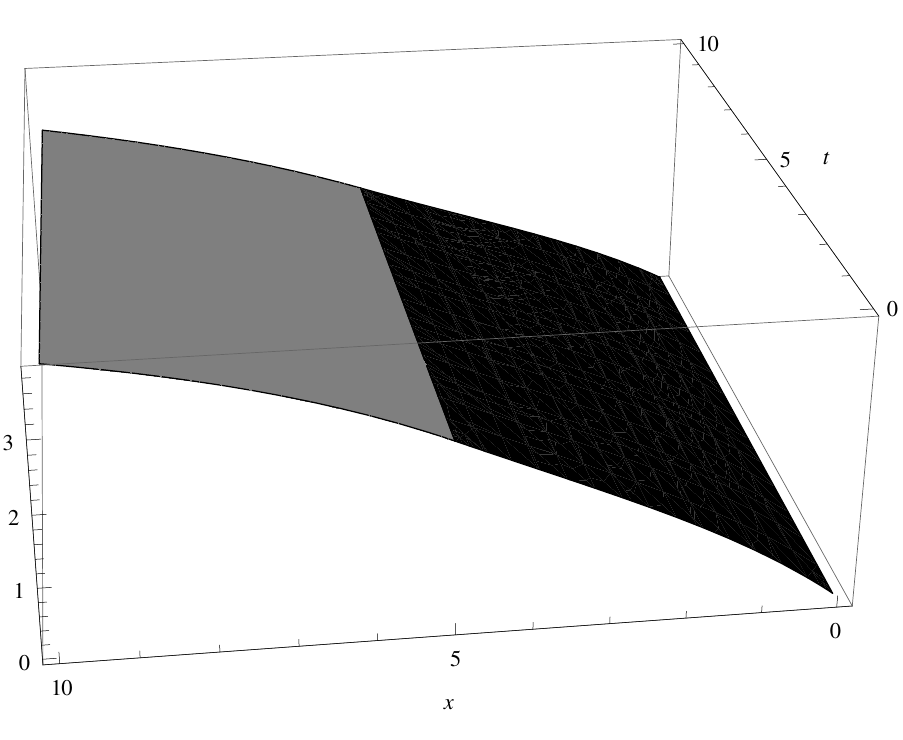}
\includegraphics[scale=0.6, bb = -140 0 100 200]{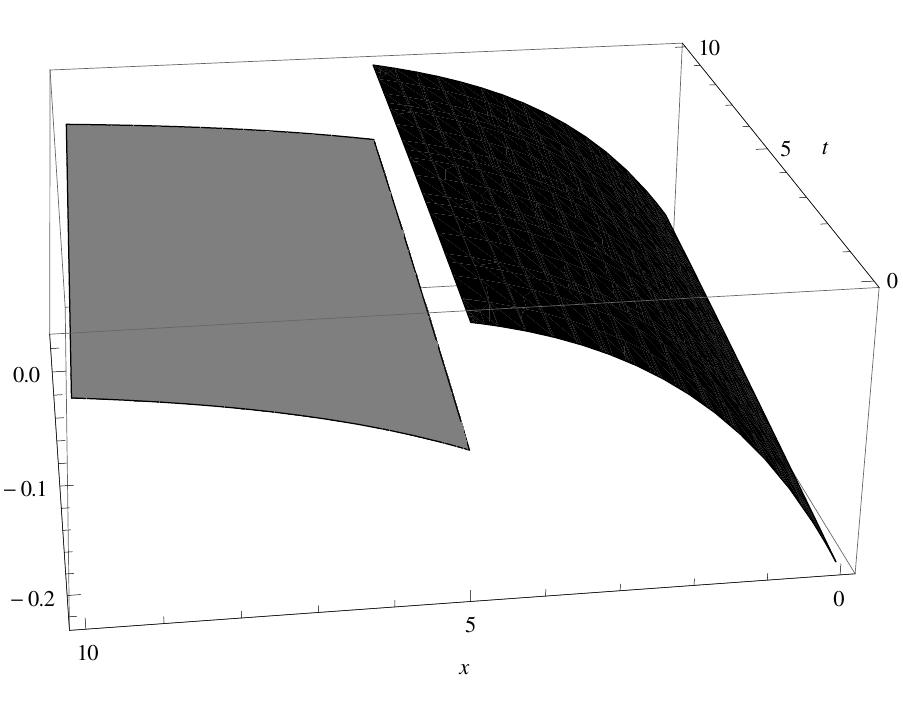}
\caption{The return function corresponding to a $5$-barrier strategy and its second derivative with respect to $x$.\label{const}}
\end{figure}
In order to analyse the performance function of a barrier strategy we will develop the  performance function into integer powers of $(e^{-\delta t})$ with $x$-dependent coefficients and truncate at some $N$. This will result in an approximation for the performance function which is much easier to handle but this incurs an additional truncation error. Inspecting Equations \eqref{F2}, \eqref{G1} motivates the approximations
 \begin{align*}
   F^N(t,x) &= \s_{n=1}^N e^{-\delta tn} \s_{k=0}^n A_{n,k} e^{\eta_k (x-q)}, \\
   G^N(t,x) &:= \s_{n=1}^N D_n e^{-\delta tn}\frac{e^{\theta_nx}-e^{\zeta_n x}}{e^{\theta_nq}-e^{\zeta_n q}},
\end{align*}  
for $x,t\geq 0$ where $\eta_0:=0$. In order to achieve a $\mathcal C^{(1,1)}$ fit we choose $D_n := \s_{k=0}^nA_{n,k}$ and 
 $$ A_{n,n} := \frac{\s_{k=0}^{n-1}(\nu_n-\eta_k)A_{n,k}}{\eta_n-\nu_n},\quad \nu_n := \frac{\theta_ne^{\theta_nq}-\zeta_ne^{\zeta_n q}}{e^{\theta_nq}-e^{\zeta_n q}}.$$
 This leaves the choice for $A_{n,0},\dots,A_{n,k-1}$ open which we now motivate by inspecting the dynamics equation for $F,G$ which should be:
\begin{align*}
&G_t(t,x)+\mu G_x(t,x)+\frac{\sigma^2}{2}G_{xx}(t,x)  = 0, \\
&F_t(t,x)+\mu F_x(t,x)+\frac{\sigma^2}{2}F_{xx}(t,x) = \xi\Big(F_x(t,x)+e^{-\delta t}(\gamma F(t,x)-1)\Big)
\end{align*}
with boundary condition $G(t,0) = 0$ for $t,x\geq 0$. 
\\It is easy to verify that $G^N(t,0) = 0$ and $G^N_t(t,x)+\mu G^N_x(t,x)+\frac{\sigma^2}{2}G^N_{xx}(t,x) = 0$. However, since $H_k(x):=e^{\eta_k x}$ solves the equation
$$ \delta k H_k(x) = (\mu-\xi) \partial_xH_k(x)+\frac{\sigma^2}{2}\partial_x^2H_k(x) $$
we find that
\begin{align*}
F^N_t(t,x)+(\mu-\xi) F^N_x(t,x)+\frac{\sigma^2}{2}F^N_{xx}(t,x) &= \s_{n=1}^N e^{-\delta t n} \s_{k=0}^{n-1} \delta(k-n)A_{n,k}e^{\eta_k (x-q)}, 
\\e^{-\delta t}\xi(\gamma F^N(t,x)-1) &= -e^{-\delta t}\xi + \s_{n=2}^{N+1} e^{-\delta tn}\s_{k=0}^{n-1}\gamma\xi A_{n-1,k}e^{\eta_k(x-q)}
\end{align*}   
We will treat the term $e^{-\delta t(N+1)}\xi\gamma\s_{k=0}^NA_{N,k}e^{\eta_k(x-q)}$ as an error term and otherwise equate the two expressions above. This allows to define the remaining coefficients which are given by:
\begin{align*}
  A_{n,k} &:= \frac{\gamma\xi A_{n-1,k}}{\delta(k-n)} = (-\frac{\gamma\xi}{\delta})^{n-k}\frac{A_{k,k}}{(n-k)!} = (-\Delta)^{n-k}\frac{A_{k,k}}{(n-k)!}, \\
  %A_{1,0} &:= \frac{\xi}{\delta}, \\
  A_{n,0} &:= \Big(-\frac{\gamma\xi}{\delta}\Big)^{n-1}\frac{\xi}{\delta n!} = \frac{(-\gamma)^{n-1}\xi^n}{\delta^n n!} = \frac{(-\Delta)^n}{-\gamma n!}
\end{align*}
for $n\geq k\geq 1$ and the last line also for $n=0$.
\\
The following lemma shows that $F^N$ solves ``almost'' the same equation as $F$ is thought to solve. Instead of being zero we see an error term which converges for time to infinity faster than $e^{-\delta tN}$.
\begin{Lem}\label{l:almost HJB}
  We have
 \begin{align*}
 G^N_t(t,x)+\mu G^N_x(t,x)+\frac{\sigma^2}{2}G^N_{xx}(t,x)  &= 0, \\
 F^N_t(t,x)+\mu F^N_x(t,x)+\frac{\sigma^2}{2}F^N_{xx}(t,x) +\xi\psi^N(e^{-\delta t},x) &= -e^{-\delta t(N+1)}\xi\gamma\s_{k=0}^NA_{N,k}e^{\eta_k(x-q)},
\end{align*}  
for any $t\geq 0,x\geq q$ where
 $$ \psi^N(e^{-\delta t},x) := -F_x^N(t,x)+e^{-\delta t}(1-\gamma F^N(t,x)). $$
\end{Lem}
\begin{proof}
The claim follows by inserting the definitions of $G^N$ and $F^N$.
\end{proof}
We define 
\begin{align*}
&V^N(t,x) := \one_{\{x\geq q\}} F^N(t,x) + \one_{\{x<q\}} G^N(t,x),
\\&\psi^N(e^{-\delta t},x) := -V_x^N(t,x)+e^{-\delta t}(1-\gamma V^N(t,x))
\end{align*}
for any $t,x\geq 0$. We now want to compare the approximate performance function $V^N$ for the barrier strategy with level $q$ to the unknown value function. We will employ the same method as in Section \ref{s:xi good} and rely on the occupation bounds from Theorem \ref{t:occupation bound}. We have in mind that $V^N \approx V^C \leq V$.  The three error terms appearing on the right-hand side of the following proposition are in this order the error for behaving suboptimal above the barrier, the error for behaving suboptimal below the barrier and the approximation error.
\begin{Prop}\label{p:goodness barrier}
We have
\begin{align*}
&V(t,x) \le V^N(t,x) 
\\& {}+ \s_{n=1}^{N+1} e^{-\delta tn} \xi\Bigg[\int_q^\infty \bigg(\one_{\{n\neq N+1\}}\s_{k=0}^n\eta_k A_{n,k}e^{\eta_k(y-q)}+\gamma\s_{k=0}^{n-1}A_{n-1,k}e^{\eta_l(y-q)}\bigg)^+ f_n(x,y) \md y \\
  & {}+ \int_0^q \left(-D_n\frac{\theta_ne^{\theta_n y}-\zeta_ne^{\zeta_ny}}{e^{\theta_n q}-e^{\zeta_nq}} +\Big(\one_{n=1}-\gamma\one_{n\neq 1}D_n\frac{e^{\theta_n y}-e^{\zeta_ny}}{e^{\theta_n q}-e^{\zeta_nq}}\Big)\right)^+ f_n(x,y) \md y \Bigg] \\
 &{} + e^{-\delta t(N+1)}\xi\gamma \int_q^\infty \s_{k=0}^N|A_{N,k}|e^{\eta_k(y-q)} f_{N+1}(x,y)\md y
\end{align*}
for any $t,x\ge 0$ where $f_k$ are defined in Proposition \ref{p:Goodness constant}.
\end{Prop}
\begin{proof}
Observe that $V^N$ is analytic outside the barrier $q$ and $\mathcal C^{(1,\infty)}$ on $\mathbb R_+\times\mathbb R_+$ and the second space derivative is a bounded function. Thus, we can apply the change of variables formula, confer \cite{peskir}.
\\Choose an arbitrary strategy $\bar C$ and denote its ruin time by $\tau$. 
Like before, applying Lemma \ref{l:almost HJB}, taking expectations and letting $s\rightarrow\infty$ yields:
\begin{align*}
&V^{\bar C}(t,x) = V^N(t,x) + \E_{(t,x)}\Big[\int_t^{\tau} e^{-\gamma\int_t^{r} \bar C_u\md u}(\bar C_r-\xi\one_{\{X^{\bar C}_r>q\}})\psi^N(e^{-\delta r},X^{\bar C}_r)\md r\Big] 
\\&{}- \xi\gamma\E_{(t,x)}\Big[\int_t^\tau e^{-\gamma\int_t^{r} \bar C_u\md u}e^{-\delta r(N+1)}\one_{\{X^{\bar C}_r>q\}}\s_{k=0}^NA_{N,k}e^{\eta_k(X^C_r-q)}\md r\Big] 
\\&\le V^N(t,x)+ \xi\gamma\E_{(t,x)}\Big[\int_t^\tau e^{-\delta r(N+1)}\one_{\{X^{\bar C}_r>q\}}\s_{k=0}^N|A_{N,k}|e^{\eta_k(X^C_r-q)}\md r\Big]
\\& {}+ \E_{(t,x)}\Big[\int_t^{\tau} \Big(-\xi\one_{\{X^{\bar C}_r>q,\psi^N(e^{-\delta r},X^{\bar C}_r)<0\}} + \xi\one_{\{X^{\bar C}_r<q,\psi^N(e^{-\delta r},X^{\bar C}_r)>0\}}\Big)\psi^N(e^{-\delta r},X^{\bar C}_r)\md r\Big],
\end{align*}
where we used that $0\le \bar C_r\le \xi$. Inserting the definition of $\psi^N$, pulling out the sum and applying Theorem \ref{t:occupation bound} yields

\begin{align*}
&V^{\bar C}(t,x) \le V^N(t,x) 
\\&{}+ \s_{k=1}^n e^{-\delta tk} \xi\Bigg[\int_q^\infty \bigg(\s_{l=0}^k\eta_l A_{k,l}e^{\eta_l(y-q)}-\one_{k=1}+\one_{k\neq 1}\s_{l=0}^{k-1}A_{k-1,l}e^{\eta_l(y-q)}\bigg)^+ f_k(x,y) \md y \\
  &{}+\int_0^q \bigg(-D_k\frac{\theta_ke^{\theta_k y}-\zeta_ke^{\zeta_ky}}{e^{\theta_k q}-e^{\zeta_kq}} +\Big\{\one_{k=1}-\one_{k\neq 1}D_k\frac{e^{\theta_k y}-e^{\zeta_ky}}{e^{\theta_k q}-e^{\zeta_kq}}\Big\}\bigg)^+ f_k(x,y) \md y \Bigg] \\
 &{}+ e^{-\delta t(N+1)}\xi\gamma \int_q^\infty \s_{k=0}^N|A_{N,k}|e^{\eta_k(y-q)} f_{N+1}(x,y)\md y.
\end{align*}
Since $\bar C$ was an arbitrary strategy and the right hand side does not depend on $\bar C$, the claim follows.
\end{proof}
Now we quantify the notion $V^N\approx V^C$. Here, we see a single error term which corresponds to the approximation error (third summand) in Proposition \ref{p:goodness barrier}.
\begin{Lem}
Let $t,x\geq 0$. Then we have
\[
|V^N(t,x)-V^C(t,x)| \le e^{-\delta t(N+1)} \xi\gamma \int_q^\infty \s_{k=0}^N|A_{N,k}|e^{\eta_k(y-q)} f_{N+1}(x,y)\md y.
\]
\end{Lem}
\begin{proof}
By following the lines of the proof of Proposition \ref{p:goodness barrier} with the specific strategy $\bar C_t=C_t = \xi \one_{\{X_t^C>q\}}$ until estimates are used yields

\begin{align*}
V^C(t,x) &= V^N(t,x)  + \E_{(t,x)}\Big[\int_t^{\tau} e^{-\gamma\int_t^{r}  C_u\md u}( C_r-\xi\one_{\{X^{ C}_r>q\}})\psi^N(r,X^{ C}_r)\md r\Big] \\
  &\quad - \xi\gamma\E_{(t,x)}\Big[\int_t^\tau e^{-\gamma\int_t^{r}  C_u\md u}e^{-\delta r(N+1)}\one_{\{X^{ C}_r>q\}}\s_{k=0}^NA_{N,k}e^{\eta_k(X^C_r-q)}\md r\Big] \\
  &\quad - \xi\gamma\E_{(t,x)}\Big[\int_t^\tau e^{-\gamma\int_t^{r}  C_u\md u}e^{-\delta r(N+1)}\one_{\{X^{ C}_r>q\}}\s_{k=0}^NA_{N,k}e^{\eta_k(X^C_r-q)}\md r\Big].
\end{align*}
Trivial inequalities and Theorem \ref{t:occupation bound} yield the claim.
\end{proof}

\section{Examples \label{bsp}}
\noindent
Here, we consider two examples. The first one will illustrate how the value function and the optimal strategy can be calculated using a straightforward approach under various unproven assumptions. In fact, we will assume (without proof) that the value function is smooth enough, the optimal strategy is of barrier type and that the barrier, the value function above the barrier and the value function below the barrier have suitable power series representations. In \cite{ghsz} has been observed that similar power series -- if exist -- have very large coefficients for certain parameter choices. This could mean that the power series doesn't converge or that insufficient computing power was at hand.
\\
In the second subsection, we will illustrate the new approach and calculate the distance of the performance function of a constant barrier strategy to the value function. The key advantages of this approach are that we do not rely on properties of the value function, nor do we need to know how it looks like. From a practical perspective, if the value function cannot be found, one should simply choose any strategy with an easy-to-calculate return function. Then, it is good to know how large the error to the optimal strategy is.
\subsection{The straightforward approach}
\noindent
In this example we let $\mu = 0.15$, $\delta = 0.05$, $\gamma = 0.2$ and $\sigma=1$. We try to find the value function numerically. However, we do not know whether the assumptions which we will make do actually hold true for any possible parameters --- or, even for the parameters we chose. 

\noindent
We conjecture and assume that the optimal strategy is of a barrier type where the barrier is given by a time-dependent curve, say $\alpha$; the value function $V(t,x)$ is assumed to be a $\mathcal C^{1,2}(\R_+^2)$ function and we define
\begin{align*}
h(t,x) &:= V(t,x),\quad t\geq 0, x\in[\alpha(t),\infty), \\
g(t,x) &:= V(t,x),\quad t\geq 0, x\in[0,\alpha(t)],
\end{align*}
\begin{figure}[t]
\begin{center}\includegraphics[scale=0.4,]{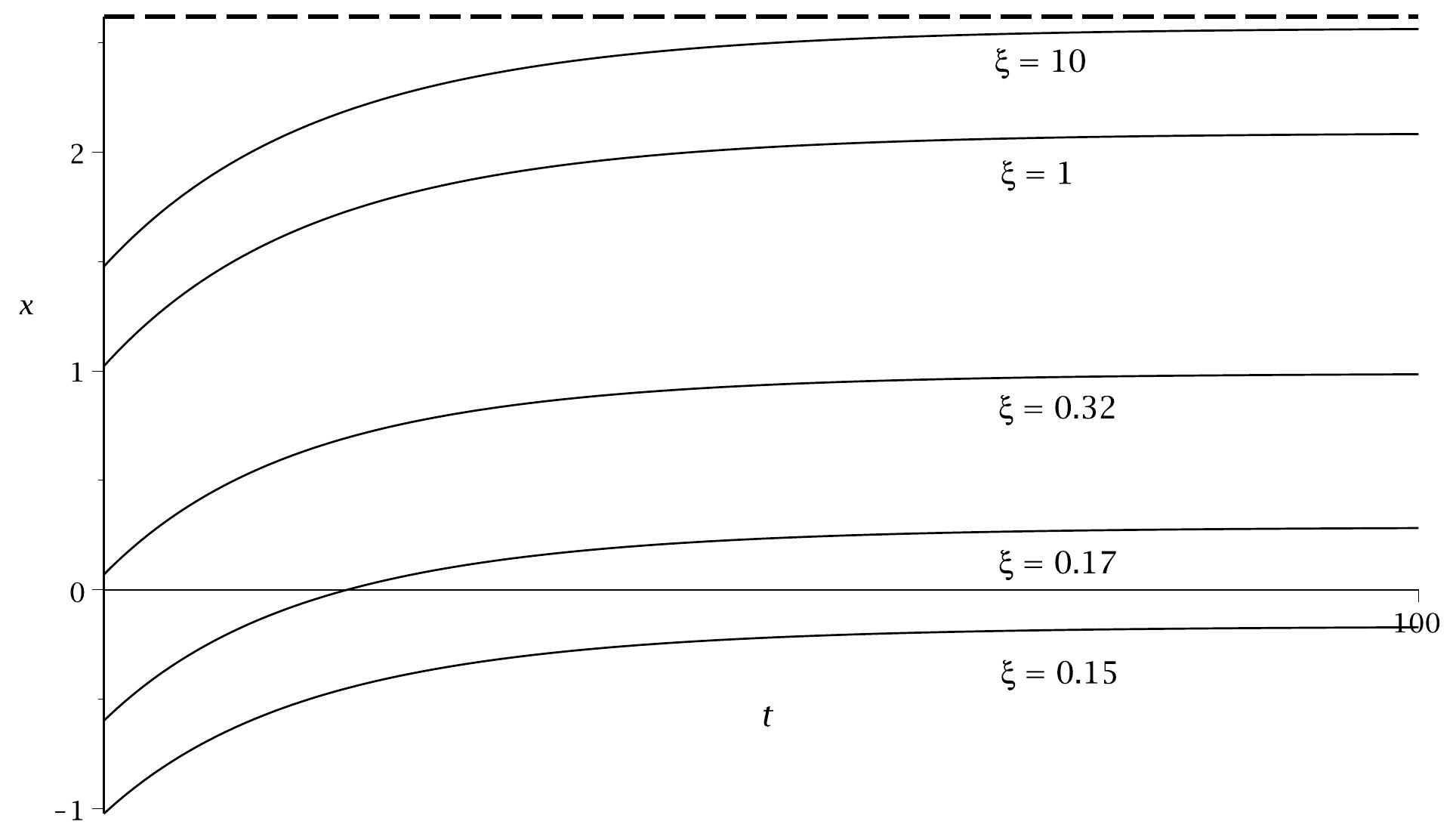}\end{center}
\caption{The optimal strategies for different values of $\xi$. The dashed line corresponds to the Asmussen-Taksar strategy \cite{astak} (unrestricted dividend case).}
\end{figure}
We assume that
\begin{align*}
&h(t,x):=\frac 1\gamma-\frac 1\gamma e^{-\de e^{-\delta t}}+ e^{-\de e^{-\delta t}}\s_{n=1}^\infty J_n e^{-\delta tn} e^{\eta_n x}\;,
\\&g(t,x):=\s_{n=1}^\infty L_ne^{-\delta tn} \big(e^{\theta_n x}-e^{\zeta_n x}\big)\;,
\\&\alpha(t):=\s_{n=0}^\infty \frac{a_n}{n!}e^{-\delta tn}\;,
\end{align*}
for some coefficients. Note that we do not investigate the question whether the functions $h$, $g$ and $\alpha$ have a power series representation. We define further auxiliary coefficients $b_{k,n},p_{k,n}$ and $q_{k,n}$:
\begin{align*}
e^{\eta_n \alpha(t)}=:\s_{k=0}^\infty\frac{b_{k,n}}{k!}e^{-\delta tk},
\quad\quad e^{\theta_n \alpha(t)}=:\s_{k=0}^\infty\frac{p_{k,n}}{k!}e^{-\delta tk},
\quad\quad e^{\zeta_n \alpha(t)}=:\s_{k=0}^\infty\frac{q_{k,n}}{k!}e^{-\delta tk}\;,
\end{align*}
Since we assume that the value function is twice continuously differentiable with respect to $x$ we have
\begin{align}
&h(t,\alpha(t))=g(t,\alpha(t)),\nonumber
\\&g_x(t,\alpha(t))=h_x(t,\alpha(t)), \label{smooth}
\\&g_{xx}(t,\alpha(t))=h_{xx}(t,\alpha(t)),\nonumber
%\\&h_{xx}(t,\alpha(t))=g_{xx}(t,\alpha(t))\;.\nonumber
%\\&g_x(t,\alpha(t))=e^{-\delta t}\Big(1-\gamma g(t,\alpha(t))\Big), \label{smooth}
%\\&h_x(t,\alpha(t))=e^{-\delta t}\Big(1-\gamma h(t,\alpha(t))\Big).\nonumber
%%\\&h_{xx}(t,\alpha(t))=g_{xx}(t,\alpha(t))\;.\nonumber
\end{align}
\begin{figure}[t]
\includegraphics[scale=0.65, bb =-30 20 200 150]{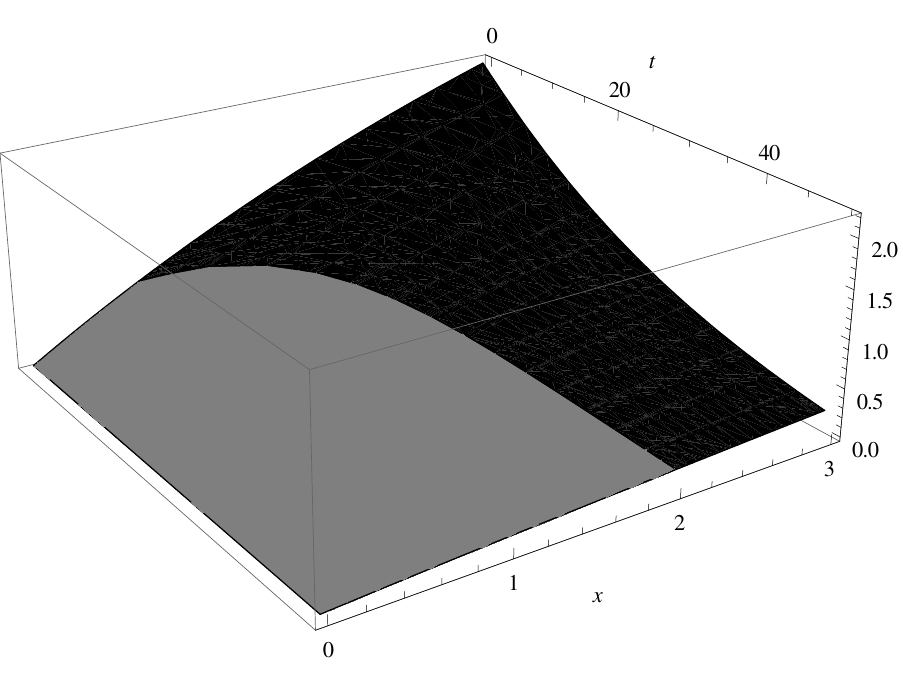}
\includegraphics[scale=0.65, bb =-130 20 200 150]{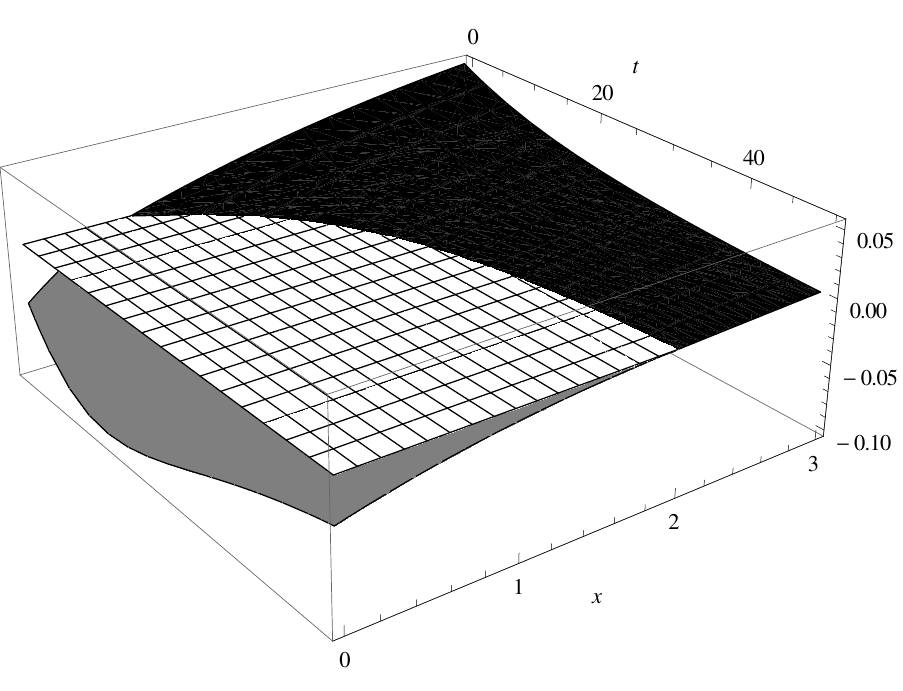}
\caption{The functions $h(t,x)$ (black) \& $g(t,x)$ (gray) in the left picture and the functions $-h_x+e^{-\delta t}(1-\gamma h)$ (black) \& $-g_x+e^{-\delta t}(1-\gamma g)$ (gray) \& $0$ (white) in the right picture for $\xi=1$.}
\end{figure}
\begin{figure}[h]
\includegraphics[scale=0.55, bb =120 50 500 750]{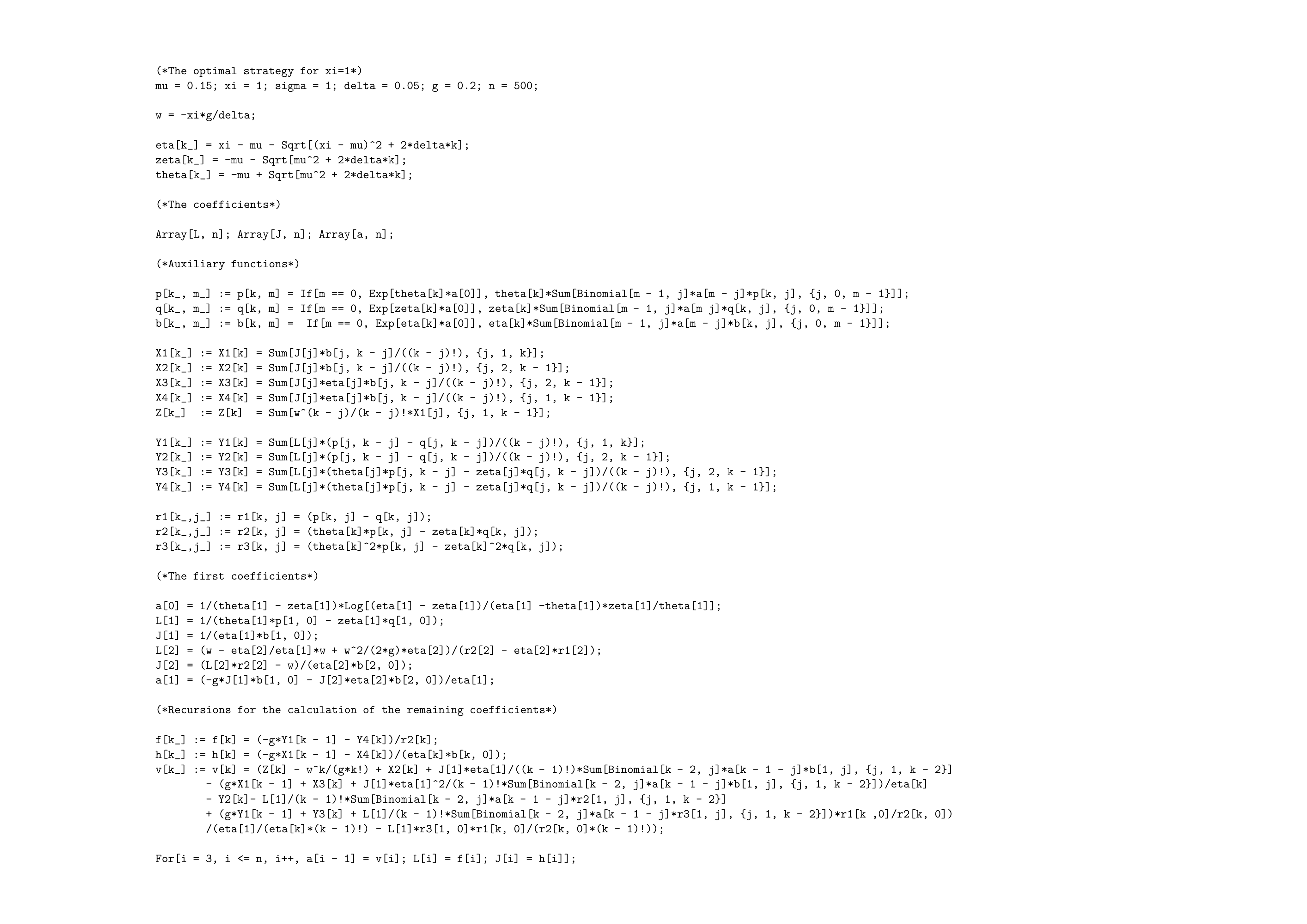}
\caption{Mathematica code for the calculation of the coefficients $J_n$, $L_n$ and $a_n$.}
\end{figure}
Note that \eqref{smooth} yields $h_t(t,\alpha(t))=g_t(t,\alpha(t))$. Therefore, we can conclude $h_x(t,\alpha(t))=e^{-\delta t}\Big(1-\gamma h(t,\alpha(t))\Big)$.
Thus, we can find the coefficients $a_n$, $J_n$ and $L_n$ from the three equations \eqref{smooth}.
Note that using the general Leibniz rule, one gets
\begin{align*}
&b_{k+1,n}=\eta_n\s_{j=0}^k\binom{k}{j}a_{k-j+1}b_{j,n},\quad b_{0,n}=e^{\eta_n\alpha(0)}, 
\\&p_{k+1,n}=\theta_n\s_{j=0}^k\binom{k}{j}a_{k-j+1}p_{j,n},\quad p_{0,n}=e^{\theta_n\alpha(0)}, 
\\&q_{k+1,n}=\zeta_n\s_{j=0}^k\binom{k}{j}a_{k-j+1}q_{j,n},\quad q_{0,n}=e^{\zeta_n\alpha(0)}\;.
\end{align*}
For $m\in\{0,1,2,3\}$ we define the coefficients
\begin{align*}
&X_{m,j}:=\s_{n=1}^jJ_n\eta_n^{m-1}\frac{b_{j-n,n}}{(j-n)!},  &&  Z_{m,k}:=\s_{j=1}^k \frac{\de^{k-j}}{(k-j)!}X_{m,j},\\
&W_{m,k,j}:=L_j\big(\theta_j^{m-1}p_{k,j}-\zeta_j^{m-1}q_{k,j}\big),  && Y_{m,k}:=\s_{n=1}^k \frac{W_{m,k-n,n}}{(k-n)!}.
\end{align*}
\bigskip
Thus, we have
\begin{align*}
&g(t,\alpha(t))=\s_{k=1}^\infty e^{-\delta tk}Y_{1,k},  && h(t,\alpha(t))=\s_{k=1}^\infty e^{-\delta tk} Z_{1,k}-\frac1\gamma\s_{k=1}^\infty \Big(-\frac{\xi\gamma}{\delta}\Big)^k\frac{e^{-\delta tk}}{k!},
\\&g_x(t,\alpha(t))=\s_{k=1}^\infty e^{-\delta tk}Y_{2,k}, && h_x(t,\alpha(t))=\s_{k=1}^\infty e^{-\delta tk}Z_{2,k},
\\&g_{xx}(t,\alpha(t))=\s_{k=1}^\infty e^{-\delta tk}Y_{3,k}, && h_{xx}(t,\alpha(t))=\s_{k=1}^\infty e^{-\delta tk}Z_{3,k}.
%\\&=-\frac{2\mu}{\sigma^2}g_x(t,\alpha(t))-\frac{2}{\sigma^2}\s_{k=1}^\infty t^k\s_{n=1}^k nL_n\frac{p_{k-n,n}- q_{k-n,n}}{(k-n)!}
%h(t,\alpha(t))&=\s_{k=1}^\infty t^kZ_{1,k}-\frac1\gamma\s_{k=1}^\infty \Big(-\frac{\xi\gamma}{\delta}\Big)^k\frac{t^k}{k!}
%\\h_x(t,\alpha(t))&=\s_{k=1}^\infty t^kZ_{2,k}
%\\h_{xx}(t,\alpha(t))&=\s_{k=1}^\infty t^kZ_{3,k}
%\s_{n=1}^\infty L_nt^n \big(e^{\theta_n x}-e^{\zeta_n x}\big)=
\end{align*}
Equating coefficients yields
\begin{align}
Y_{1,k}=Z_{1,k}-(-1)^k\frac{\xi^k\gamma^{k-1}}{\delta^kk!},\quad\quad\quad
Y_{2,k}=Z_{2,k},\quad\quad\quad
 Y_{3,k}=Z_{3,k}\;. \label{fg}
\end{align}
Note that Equations \eqref{fg} specify $L_k$, $J_k$ and $a_{k-1}$ in $k$th step.
The coefficients given above have a recursive structure. Due to this fact the numerical calculations turn out to be very time- and disk space-consuming. 
Numerical calculations show that the above procedure yields well-defined power series for relative small values of $\xi$. However, for big $\xi$ the coefficients explode, which makes the calculations unstable and imprecise especially for $t$ close to zero.
\subsection{The distance to the value function}\label{S:approx}
\begin{figure}[t]\label{F:4plots}
\begin{center}
\includegraphics[scale=.7,bb = 250 50 300 430]{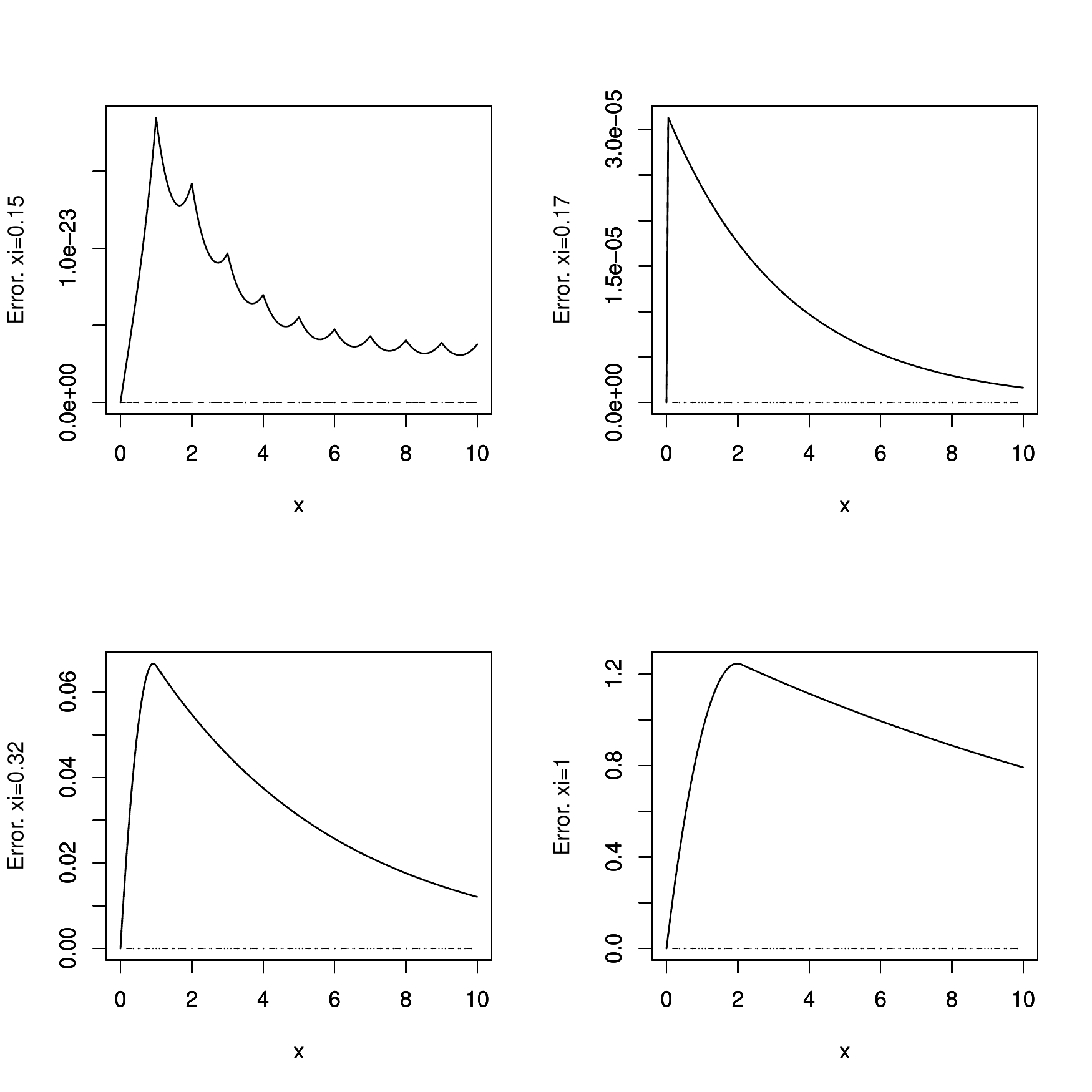}
\end{center}
\caption{The plots show the error bounds given by Proposition \ref{p:goodness barrier} for the barrier strategy with parameters given in Section \ref{S:approx} at time $t=0$.}
\end{figure}
We use the same parameters as in the previous section, i.e.\ $\mu = 0.15$, $\delta = 0.05$, $\gamma = 0.2$ and $\sigma=1$. We illustrate the error bound given by Proposition \ref{p:goodness barrier} for $N=20$ summands and four different values for $\xi$, namely $0.15$, $0.17$, $0.32$ and $1$. We will compare the unknown value function to the performance of the barrier strategy with barrier at
 $$ q = \bigg(\frac{\log(-\zeta_1)+\log(\zeta_1+\eta_1)-\log(\theta_1)-\log(\theta_1-\eta_1)}{\theta_1-\zeta_1}\bigg)^+, $$
 i.e.\ we employ the strategy $C_s = \xi\one_{\{X_s^C\geq q\}}$. This barrier strategy has been shown to be optimal if no utility function is applied, confer \cite[p.\ 97]{HS}. In the case of $\xi=0.15$ one finds $q=0$, i.e.\ we pay out at maximal rate all the time which is optimal due to Proposition \ref{p:0 opt}. Therefore, this case is left with approximation error only. For the other values of $\xi$, it is non-optimal to follow a barrier strategy and, hence, we do have a substantial error which cannot disappear in the limit. The corresponding pictures in Figure \ref{F:4plots} show this error as for $N=20$ summands the approximation error is already several magnitudes smaller than the error incurred by following a suboptimal strategy.

\appendix
\section{Appendix}\label{s:Appendix}
 In this section we provide deterministic upper bounds for the expected discounted occupation of a process whose drift is not precisely known. This allows to derive an upper bound for the expect discounted and cumulated positive functional of the process. These bounds are summarised in Theorem \ref{t:occupation bound}.
\\Let $a,b\in\mathbb R$ with $a\leq b$, $I:=[a,b]$, $\sigma>0$, $\delta\geq 0$, $W$ a standard Brownian motion and consider the process
  $$ \md X_t = C_t \md t + \sigma \md W_t $$
where $C$ is some $I$-valued progressively measurable process. We recall that we denote by $\mP_x$ a measure with $\mP_x[X_0=x]$. The local time of $X$ at level $y$ and time $t$ is denoted by $L_t^y$ and $\tau:=\inf\{t\geq 0: X_t = 0\}$. Further we define for $x,y\geq 0$
\begin{align*}
&\alpha := \frac{a+\sqrt{a^2+2\delta\sigma^2}}{\sigma^2}, 
 \quad\quad \beta_+ := \frac{\sqrt{b^2+2\delta\sigma^2}-b}{\sigma^2}, 
  \quad\quad  \beta_- := \frac{-\sqrt{b^2+2\delta\sigma^2}-b}{\sigma^2}, \\
      &f(x,y) := \frac{2\left(e^{\beta_+(x \wedge y)}-e^{\beta_-(x\wedge y)}\right)e^{-\alpha(x-y)^+}}{\sigma^2\left((\beta_++\alpha)e^{y\beta_+}-(\beta_-+\alpha)e^{y\beta_-}\right)}.
  \end{align*}

 \begin{Thm}\label{t:occupation bound}
We have $\E_x\left[ \int_0^\tau e^{-\delta s} L_{\md s}^y\right]  \le \sigma^2 f(x,y)$.
   In particular, for any measurable function $\psi:\mathbb R_+\rightarrow\mathbb R_+$ we have
    $$ \E_x\left[\int_0^\tau e^{-\delta s}\psi(X_s) \md s \right] \leq \int_0^\infty \psi(y)f(x,y) \md y. $$
 \end{Thm}  
The proof is given at the end of this section.  
  
 \begin{Lem}
    $f$ is absolutely continuous in its first variable with derivative
       \begin{align*}
      f_x(x,y) &:= \begin{cases} 
      \frac{2\left(\beta_+ e^{x\beta_+}-\beta_-e^{x\beta_-}\right)}{\sigma^2\left((\beta_++\alpha)e^{y\beta_+}-(\beta_-+\alpha)e^{y\beta_-}\right)} & x \le y, \\ 
      \frac{2\left(-\alpha e^{y\beta_+}-\alpha e^{y\beta_-}\right)e^{-\alpha(x-y)}}{\sigma^2\left((\beta_++\alpha)e^{y\beta_+}-(\beta_-+\alpha)e^{y\beta_-}\right)} & x>y.
      \end{cases}
  \end{align*}
   For any $y\geq 0$ the function $f_x(\cdot,y)$ is of finite variation and
      \begin{align*}
         \md f_x(x,y) &= -\frac{2}{\sigma^2}\delta_y(\md x) + \left(\frac{2\delta}{\sigma^2}f(x,y)-\frac{2(b\one_{\{x<y\}}+a\one_{\{x>y\}})}{\sigma^2}f_x(x,y)\right)\md x
      \end{align*}
     where $\delta_y$ denotes the Dirac-measure in $y$. Moreover, if we denote by $f_{xx}(x,y)$ the second derivative of $f$ with respect to the first variable for $x\neq y$, then we get
     $$ \sup_{u\in[a,b]}\left(\frac{\sigma^2}{2}f_{xx}(x,y) + uf_x(x,y)-\delta f(x,y)\right) = 0,\quad x\neq y. $$
 \end{Lem}
 \begin{proof}
   Straightforward.
 \end{proof}

 \begin{Lem}
   Let $y\geq 0$ and assume that $C_t = a\one_{\{X_t>y\}} + b\one_{\{X_t\leq y\}}$. Then
    $$ \E_x\Big[\int_0^\tau e^{-\delta s}L_{\md s}^y\Big] = \sigma^2f(x,y). $$
 \end{Lem} 
 \begin{proof}
   Ito Tanaka's formula together with the occupation time formula yield
    \begin{align*}
       f(X_{t\w \tau},y) &= f(x,y) + \int_0^t \sigma f_x(X_{s\w \tau},y) \md W_s - \frac{1}{\sigma^2}L_{t\wedge\tau}^y \\&\quad+ \int_0^t C_sf_{x}(X_{s\w \tau},y) + \frac{\sigma^2}{2}f_{xx}(X_{s\w \tau},y)\md s \\
          &= f(x,y) + \int_0^t \sigma f_x(X_{s\w \tau},y) \md W_s - \frac{1}{\sigma^2}L_{t\wedge\tau}^y + \delta \int_0^t f(X_{s\w \tau},y)\md s.
    \end{align*}
 Using the product formula yields
  $$ e^{-\delta t}    f(X_{t\w \tau},y) = f(x,y) + \int_0^t \sigma e^{-\delta s} f_x(X_{s\w \tau},y) \md W_s -\frac1{\sigma^2} \int_0^{t\wedge \tau} e^{-\delta s} \md L^y_{\md s}. $$
 Since $f_x(\cdot,y)$ is bounded we see that the second summand is a martingale. If $\delta >0$, then we find that
  $$ \lim\limits_{t\rightarrow\infty} \E_x[ e^{-\delta t}    f(X_{t\wedge\tau},y)] = 0. $$
 If $\delta = 0$ and $a\leq 0$, then $\tau < \infty$ $\mP$-a.s.\ and boundedness of $f$ yields 
  $$ \lim_{t\rightarrow\infty} \E_x[ f(X_{t\w \tau},y)] = 0. $$
 If $\delta = 0$ and $a > 0$, then $X_{t\w\tau}$ takes values in $\{0,\infty\}$ and $\lim\limits_{x\rightarrow\infty} f(x,y)=0$, thus boundedness of $f$ yields again
  $$ \lim_{t\rightarrow\infty} \E_x[ f(X_{t\w \tau},y)] = 0. $$
 Thus, we find by monotone convergence
  $$ 0 = f(x,y) -\frac1{\sigma^2} \lim_{t\rightarrow\infty} \E_x\left[\int_0^{t\wedge \tau} e^{-\delta s} \md L^y_{\md s}\right] = f(x,y) -\frac1{\sigma^2} \E_x\left[\int_0^{\tau} e^{-\delta s}  \md L^y_{\md s}\right].$$
 \end{proof}

\begin{proof}[Proof of Theorem \ref{t:occupation bound}]
   Fix $y\geq 0$. For any progressively measurable process $\eta$ with values in $I$ we define
\begin{align*}
Y^\eta_t := X_0 + \int_0^t \eta_s \md s + \sigma W_t \quad\mbox{and}\quad
       V(x) := \sup_{\eta} \E_x\left[ \int_0^{\tau} e^{-\delta s} L^{y,\eta}_{\md s}\right],
\end{align*}    
where $\tau^{\eta} := \inf\{t\geq 0: Y^\eta_t = 0\}$ and $L^{\cdot,\eta}$ denotes a continuous version of the local time of $Y^{\eta}$. Clearly, we have
 $$ \E_x\left[ \int_0^{\tau} e^{-\delta s} L^{y,\eta}_{\md s} \right] \leq V(x). $$
  Moreover, the previous two lemmas yield that $Y^{\eta^*}$ with 
   $$\eta_t^* = a\one_{\{Y^{\eta^*}_t>y\}} + b\one_{\{Y^{\eta^*}_t\leq y\}}$$ 
   is the optimally controlled process and we get $V(x) = \sigma^2f(x,y)$. (The process $\eta^*$ exists because the corresponding SDE admits pathwise uniqueness, confer \cite[Thm IX.3.5]{revuz.yor.05}.)
\end{proof} 
\section{Lower unbounded drift occupation bound}
We will generalise the results from the previous section to the case where no lower bound on the drift is given, i.e.\ the drift is only assumed to be upper Lipschitz-continuous with some rate $b$ (which might be negative). The bounds are summarised in Theorem \ref{t:occupation bound}. Let $b\in\mathbb R$, $\sigma>0$, $\delta\geq 0$, $W$ a standard Brownian motion and consider the process
  $$ \md X_t = \md C_t + \sigma \md W_t $$
 where $C$ is some $I$-valued progressively measurable process satisfying $C_t-C_s \le b(t-s)$ for any $0\leq s\leq t$. And again, we denote by $\mP_x$ a measure with $\mP_x[X_0=x]$. The local time of $X$ at level $y$ and time $t$ is denoted by $L_t^y$ and $\tau:=\inf\{t\geq 0: X_t = 0\}$.
\\Further we define for $x,y\geq 0$
  \begin{align*}
  &\beta_+ := \frac{\sqrt{b^2+2\delta\sigma^2}-b}{\sigma^2}, \quad \beta_- := \frac{-\sqrt{b^2+2\delta\sigma^2}-b}{\sigma^2}, \quad
      f(x,y) := \frac{2\left(e^{\beta_+(x \wedge y)}-e^{\beta_-(x\wedge y)}\right)}{\sigma^2\left(\beta_+e^{y\beta_+}-\beta_-e^{y\beta_-}\right)}.
  \end{align*}

 \begin{Thm}\label{t:occupation bound unbounded}
   We have $ \E_x\left[ \int_0^\tau e^{-\delta s} L_{\md s}^y\right]  \le \sigma^2 f(x,y)$.
   In particular, for any measurable function $\psi:\mathbb R_+\rightarrow\mathbb R_+$ we have
    $$ \E_x\left[\int_0^\tau e^{-\delta s}\psi(X_s) \md s \right] \leq \int_0^\infty \psi(y)f(x,y) \md y. $$
 \end{Thm}  
 \begin{proof}
   The proof follows from Theorem \ref{t:occupation bound} by taking the limit $a\rightarrow-\infty$.
 \end{proof}

\section*{Acknowledgments}
\noindent
The research of the first author was funded by the Austrian Science Fund (FWF), Project number V 603-N35. The first author is currently on leave from the University of Liverpool and would like to thank the University of Liverpool for support and cooperation.

\end{document}